\documentclass[12pt]{amsart}
\usepackage{geometry} % see geometry.pdf on how to lay out the page. There's lots.
\usepackage{amssymb}
\usepackage{amscd}
\usepackage{amsmath}
\usepackage[all,cmtip]{xy}
\usepackage{mathrsfs}
\usepackage{amsthm}
\usepackage[pdfpagelabels]{hyperref}
\newfont{\Bb}{msbm10 scaled1200}
\newfont{\bb}{msbm8 scaled1200}

\def\beq{\begin{equation}}
\def\eeq{\end{equation}}
\def\beqa{\begin{eqnarray}}
\def\eeqa{\end{eqnarray}}

\newtheorem{thm}{Theorem} [section]
\newtheorem{cor}[thm]{Corollary}
\newtheorem{lem}[thm]{Lemma}

\theoremstyle{definition}

\newtheorem{rem}[thm]{Remark}

\begin{document}

\title[A Multiparametric Quantum Superspace]{A Multiparametric Quantum Superspace and Its Logarithmic Extension}
\author[Muttalip \"{O}zav\c{s}ar ]{Muttalip \"{O}zav\c{s}ar}
\address{Department
of Mathematics, 
Yildiz Technical University 
Davutpasa-Esenler
P.O. Box 34210 
Turkey}
\email{mozavsar@yildiz.edu.tr}
\author{Erg\"{u}n Ya\c{s}ar}

\address{Department
of Mathematics 
Yildiz Technical University
Davutpasa-Esenler
P.O. Box 34210 
Turkey}
\email{eyasar@yildiz.edu.tr}

\begin{abstract}
We introduce a multiparametric quantum superspace with $m$ even generators and $n$ odd generators whose commutation relations are in the sense of Manin such that the corresponding algebra has a Hopf superalgebra. By using its Hopf superalgebra structure, we give a bicovariant differential calculus and some related structures such as Maurer-Cartan forms and the correspoinding vector fields. It is also shown that there exists a quantum supergroup related with these vector fields.
Morever, we introduce the logarithmic extension of this quantum superspace in the sense that we extend this space by the series expansion of the logarithm of the grouplike generator, and we define new elements with nonhomogeneous commutation relations. It is clearly seen that this logarithmic extension is a generalization of the $\kappa-$Minkowski superspace. We give the bicovariant differential calculus and the related algebraic structures on this extension. All noncommutative results are found to reduce to those of the standard superalgebra when the deformation parameters of the quantum (m+n)-superspace are set to one.  
\end{abstract}
\pagenumbering{arabic}
\maketitle

\section{Introduction}	
Noncommutative geometry has gained more attention of researchers as
research domain in the fields of mathematics and mathematical physics since noncommutative
differential geometry was broadly introduced by Connes \cite{AC} in 1986. In particular, quantum groups (Refs.~\cite{LDF, MJ1, MJ2, SLW1, VG, YMI1, YMI2}) and quantum spaces
(Refs.~\cite{RMU, SLW2, YMI1, YMI2}) are
explicit realizations of noncommutative spaces and play a fundamental role in the theory of the integrable models, conformal field theory \cite{GN, LCG} and the classification of knots and and links \cite{JF, RBZ}. 
The quantum (super)spaces have
been envisioned by many as a paradigm for the general programme of quantum deformed
physics \cite{SM}. Thus, based on the fact that the study of differential calculus is a main mathematical tool in the quantum deformed physics, many efforts have been accomplished in order to develop noncommutative
differential structures on quantum superspaces(groups) (Refs.~\cite{AEH, AS, CF, DPP, EMF, JW, MAY, NA, NAR, PW, RC, SC, SLW3, TB, TK, VR, ZB, IA, OO}). 
In particular, as a fundamental work,  the study of differential calculus on noncommutative space of quantum groups was initiated by Woronowicz \cite{SLW3}. In Woronowicz's approach, the differential calculus on quantum groups is inferred by Hopf algebra structure of quantum groups and this
calculus is extended to graded differential Hopf algebras. Later, Wess and Zumino introduced the differential calculus on the quantum (hyper-)plane which is covariant with respect to the quantum group \cite{JW}.

In this paper, we give a multiparametric quantum $(m+n)$-superspace on which we define a Hopf superalgebra, and its bicovariant differential calculus is given by using its  Hopf superalgebra structure. Then we define new generators with nonhomegeneous commutation relations via the series expansion of logarithm of the grouplike generator. First we recall some definitions and statements which shall be used throughout the paper.

An associative \textit{algebra} is a vector space $\mathcal{A}$
over a field $K$  together with  a bilinear mapping, namely, the multiplication $ \mu:\mathcal{A}
\otimes \mathcal{A} \rightarrow \mathcal{A}$ 
satisfying
\begin{eqnarray} \label{1}
\begin{aligned}
\mu\circ(\mbox{id}\otimes \mu)=\mu\circ( \mu\otimes \mbox{id})
\end{aligned}
\end{eqnarray}
for all $ a,b,c \in \mathcal{A}$. Morever, if there exists a mapping $\eta: K\rightarrow \mathcal{A}$ such that 
\begin{equation}\label{unit}
\mu\circ(\eta \otimes \mbox{id})=\mbox{id}=\mu\circ(\mbox{id}\otimes \eta ),
\end{equation}
where $ \mbox{id}$ stands for the identity mapping,
then $\mathcal{A}$ is a unital algebra.  A \textit{coalgebra} is a $K$-vector space $\mathcal{A}$, together with two linear mappings, $\Delta_\mathcal{A}:
 \mathcal{A}\longrightarrow \mathcal{A} \otimes \mathcal{A}$ and $\varepsilon_{\mathcal{A}}: \mathcal{A} \longrightarrow K$
(the coproduct and the counit, respectively) which satisfy
\begin{eqnarray}\label{2}
\begin{aligned}
&(\Delta_{\mathcal{A}} \otimes \mbox{id}) \circ \Delta_{\mathcal{A}} = (\mbox{id} \otimes
\Delta_{\mathcal{A}}) \circ \Delta_{\mathcal{A}} \\
& (\varepsilon_{\mathcal{A}} \otimes \mbox{id}) \circ \Delta_{\mathcal{A}} = \mbox{id}= (\mbox{id} \otimes \varepsilon_{\mathcal{A}})\circ \Delta_{\mathcal{A}}.
\end{aligned}
\end{eqnarray}
A \textit{bialgebra} is both a unital associative algebra and a coalgebra, with the compatibility
conditions that $\Delta_{\mathcal{A}}$ and $\varepsilon_{\mathcal{A}} $ are both algebra homomorphisms with $\Delta(1_{\mathcal{A}})=1_{\mathcal{A}}\otimes1_{\mathcal{A}} $
and  $\varepsilon_{\mathcal{A}}(1_{\mathcal{A}})=1_K $. A \textit{Hopf algebra} is a bialgebra together with a linear mapping $ S_{\mathcal{A}}:\mathcal{A} \longrightarrow \mathcal{A} $,
the antipode, which satisfies
\begin{eqnarray} \label{4}
\begin{aligned}
 \mu \circ (S_{\mathcal{A}} \otimes \mbox{id}) \circ \Delta_{\mathcal{A}} = \eta \circ \varepsilon_{\mathcal{A}} = \mu \circ (\mbox{id} \otimes S_{\mathcal A})\circ\Delta_{\mathcal{A}}.
\end{aligned}
\end{eqnarray}
 
Let $\Omega$ be a bimodule over any Hopf algebra $\mathcal{A}$ and
$\Delta_R : \Omega \longrightarrow \Omega \otimes \mathcal{A} $ be a
linear homomorphism. One says that $(\Omega, \Delta_R) $ is a
\textit{right-covariant bimodule} if
\begin{equation} \label{5}
\begin{aligned}
\Delta_R(ap+p'a') = \Delta_{\mathcal{A}}(a)\Delta_R(p)+ \Delta_R(p')\Delta_{\mathcal{A}}(a') 
\end{aligned}
\end{equation}
for all $ a,a' \in \mathcal{A} $ and $ p,p' \in \Omega $ and
\begin{equation} \label{6}
\begin{aligned}
&(\Delta_R \otimes \mbox{id}) \circ \Delta_R  = (\mbox{id} \otimes \Delta_{\mathcal{A}}) \circ \Delta_R, \\
&\mu \circ(\mbox{id} \otimes \varepsilon_{\mathcal{A}}) \circ \Delta_R = \mbox{id}.
\end{aligned} 
\end{equation} 
Let $\Delta_L : \Omega \longrightarrow {\mathcal{A}} \otimes \Omega $ be
a linear homomorphism. One says that $(\Omega, \Delta_L) $ is a
\textit{left-covariant bimodule} if
\begin{equation} \label{7}
\Delta_L(ap+p'a') =\Delta_{\mathcal{A}}(a)\Delta_L(p)+ \Delta_L(p')\Delta_{\mathcal{A}}(a')
\end{equation}
for all $ a,a' \in \mathcal{A} $ and $ p,p' \in \Omega $ and
\begin{equation} \label{8}
\begin{aligned}
&(\mbox{id} \otimes \Delta_L) \circ \Delta_L=(\Delta_{\mathcal{A}} \otimes \mbox{id}) \circ \Delta_L, \\
 & \mu \circ(\varepsilon_{\mathcal{A}} \otimes \mbox{id})\circ \Delta_L = \mbox{id}. 
\end{aligned} 
\end{equation}

A bicovariant bimodule over $\mathcal{A}$ is a bimodule $\Omega$ with linear mappings $\Delta_R$, $\Delta_L$ such that $(\Omega, \Delta_L) $ is the left covariant bimodule, 
$(\Omega, \Delta_R) $ is the right covariant bimodule, and
\begin{equation} \label{9}
\begin{aligned}
( \mbox{id} \otimes \Delta_R )\circ \Delta_L=(\Delta_L \otimes \mbox{id})\circ \Delta_R.
\end{aligned}
\end{equation}

Let $\mathcal{A}$ be endowed with a linear mapping ${\sf d}:{\mathcal{A}}\rightarrow\Omega$ satisfying the Leibniz rule ${\sf d}(ab)={\sf d}(a)b+a{\sf d}(b)$, and $\Omega$ is the linear span of elements $a{\sf d}b$ with $a,b \in \mathcal{A}$. Then the tuble $(\Omega, {\sf d})$ is called the first order differential calculus over $\mathcal{A}$.  The algebra of higher order differential forms (or differential graded algebra) is a $\mathbb{N}_0$-graded algebra $\Omega^{\wedge}=\bigoplus_{n\geq 0}\Omega^n$, $\Omega^0=\mathcal{A}$, with the exterior differential mapping of degree one ${\sf d}:\Omega^{\wedge}\rightarrow \Omega^{\wedge}$ satisfying ${\sf d}\circ {\sf d}:={\sf d}^2=0$ and the graded Leibniz rule ${\sf d}(w_1 \wedge w_2)={\sf d}(w_1)\wedge w_2+(-1)^n w_1 \wedge {\sf d}(w_2), \ w_1 \in \Omega^n, \ w_2 \in \Omega^{\wedge}$. The differential algebra $\Omega^{\wedge}$ has a natural Hopf algebra structure obtained by the coproduct $\hat{\Delta}=\Delta_R+\Delta_L$, where $\Delta_R({\sf d}(a))=(({\sf d} \otimes \mbox{id}) \circ \Delta_{\mathcal{A}})(a), \ 
\Delta_L({\sf d}(a))=((\mbox{id} \otimes {\sf d}) \circ \Delta_{\mathcal{A}})(a)$ (see Refs.~\cite{SLW3,TB}).

Let $z_1=1$, $z_2,...,z_{m+n}$ be arbitrary integers and $p_1=1$, $p_2,...,p_{m+n}$ be nonzero complex numbers. Consider a unital associative superalgebra generated by even generators $a_i,i=1,2,...,m$ and odd ones $a_i,i=m+1,...,m+n$ satisfying the commutation relations as follows:
\begin{equation} \label{10}
a_ia_j = (-1)^{\hat{i}\hat{j}}p_j^{z_i}p_i^{-z_j} a_ja_i,
\end{equation}
for $i,j=1,2,...,m+n$. Note that $\hat{i}\in \mathbb{Z}_2$ stands for the parity of the generator $a_i$, and it follows from (\ref{10}) that $a_i^2=0$ for $\hat{i}=1$ . Throughout this paper, we denote by ${\mathcal{A}}$ this quantum superalgebra. At this position, we remark that ${\mathcal{A}}$ is a generalization of the quantum superalgebra considered in Ref.~\cite{MO1}. It is clear that this superalgebra has a Hopf superalgebra structure with the following mappings:

\begin{equation} \label{11}
\begin{aligned}
\Delta(a_i)&  =  a_1^{z_i} \otimes a_i + a_i \otimes a_1^{z_i},\quad \Delta(a_1) =  a_1 \otimes a_1  \\
\varepsilon(a_i)&= 0, \quad \varepsilon(a_1)=1 \\
S(a_i)&=-a_1^{-z_i}a_ia_1^{-z_i},\quad S(a_1)=a_1^{-1},
\end{aligned}
\end{equation}
where $i=2,3,...,m+n$.
Note that the tensor product in ${\mathcal{A}}\otimes {\mathcal{A}}$ is given as follows:
\begin{equation} \label{12}
(f\otimes u)(v\otimes g)=(-1)^{\hat{u}\hat{v}}fv\otimes ug, \ f,g,u,v \in {\mathcal{A}}.
\end{equation}
\begin{rem}
By a brief survey on more general algebras such as $\Gamma-$graded and color(Lie) algebras \cite{MS} and Lie $\tau -$algebras \cite{DG,VK,VKK}, one can see that the commutation relation given by (\ref{10}) can be represented by a bicharacter $\alpha:\Gamma \times \Gamma\rightarrow k^* $ where $\Gamma$ is an abelian group and $\alpha$ holds the conditions $\alpha(f.g,h)=\alpha(f,h)\alpha(g,h)$ and $\alpha(f,g.h)=\alpha(f,g)\alpha(f,h)$. Furthermore, one can define the multiplication in the tensor product $\mathcal{A}\otimes \mathcal{A}$ by 
\begin{equation}\label{nm}
(a_i\otimes a_j)(a_k\otimes a_l)=\alpha(g_j,g_k)a_ia_k\otimes a_ja_l
\end{equation}
 where $\alpha(g_i,g_j)=(-1)^{\hat{i}\hat{j}}p_j^{z_i}p_i^{-z_j}$ and $g_i$'s are the elements in $\Gamma$ which as degree correspond to the elements $a_i$'s. Thus this basic multiplication can be extended to the all homogeneous elements(ordered monomials) by using the properties of $\alpha$. However, when we take into account the coproduct defined in (\ref{11}) together with the multiplication (\ref{nm}), we easily see that the coproduct does not preserve the commutation relation (\ref{10}). This discrepancy is caused by the existence of $a_1^{z_i}$ in the coproduct (\ref{10}) and the fact that the generator $a_1$ is a grouplike element. Thus we can overcome from this inconsistency by setting all generators $a_i$'s as the primitive elements, that is, $\Delta(a_i)=a_i\otimes 1+ 1\otimes a_i$.  
\end{rem}

The differential algebra $\Omega^{\wedge}$ of all differential forms over ${\mathcal{A}}$ can be given by the following relations of the generators $a_i$'s with their differentials ${\sf d}(a_i) $'s:
\begin{equation} \label{13}
a_i {\sf d}(a_j) = (-1)^{\hat{i}(\hat{j}+1)}p_j^{z_i}p_i^{-z_j}{\sf d}(a_j) a_i,  
 \end{equation}
and the relations among differentials
\begin{equation} \label{14}
 {\sf d}(a_i)\wedge{\sf d}(a_j)= (-1)^{(\hat{i}+1)(\hat{j}+1)}p_j^{z_i}p_i^{-z_j}{\sf d}(a_j)\wedge {\sf d}(a_i),
\end{equation}
where ${\sf d}:\Omega^{\wedge} \rightarrow \Omega^{\wedge}$ is the exterior differential operator satisfying
\begin{equation} \label{15}
{\sf d}^2 = 0,
\end{equation}
and the graded Leibniz rule
\begin{equation} \label{16}
{\sf d}(u\wedge v) = ({\sf d} u)\wedge v + (-1)^{\hat{u}}u \wedge({\sf d} v).
\end{equation}
Note that the parity of ${\sf d}(w)$ is given as $\hat{w}+1$ for $ w\in \Omega^n$ $(n=0,1,2,...,)$, that is, ${\sf d}$ is of degree one.  Morever, one can give the right covariant  bimodule structure on the space of 1-forms via $\Delta_R:\Omega^1\rightarrow \Omega^1\otimes {\mathcal{A}}$, defined by $\Delta_R({\sf d}(f))=(({\sf d}\otimes \mbox{id})\circ \Delta)(f)$, and the left covariant bimodule one by $\Delta_L:\Omega^1\rightarrow {\mathcal{A}}\otimes\Omega^1$, defined by $\Delta_L({\sf d}(f))=((\mbox{id}\otimes{\sf d}) \circ \Delta)(f)$ for $f\in{\mathcal{A}}$, where we use the graded tensor product of mappings. This bimodule structure is also extended to the space of all higher-order forms $\Omega^{\wedge}$. 
Therefore, if we set ${\sf d}(f)=\sum_{k=1}^{m+n}{{\sf d}(a_k)\partial_{a_k}(f)}$, from the differential calculus above, it follows that the Weyl superalgebra corresponding to ${\mathcal{A}}$ is given by the relations (\ref{10}), the following relations of the derivative operators with the generators with $a_i$'s:
\begin{equation} \label{17}
\partial_{a_i}a_j=\delta_{ij}+(-1)^{\hat{i}\hat{j}}p_j^{-z_i}p_i^{z_j}a_j\partial_{a_i}, \quad i,j=1,2,...,m+n
\end{equation}
and the relations among the derivative operators
\begin{equation} \label{18}
\partial_{a_i} \partial_{a_j} = (-1)^{\hat{i}\hat{j}}p_j^{z_i}p_i^{-z_j} \partial_{a_j} \partial_{a_i},
\end{equation}
where  $\partial_{a_i}:{\mathcal{A}}\rightarrow {\mathcal{A}}$ is a linear operator acting on a monomial ordered of the form $f=a_1^{k_1}a_2^{k_2}...a_{m+n}^{k_{m+n}}$ as follows:
\begin{equation}\label{19}
 \partial_{a_i}(a_1^{k_1}a_2^{k_2}\cdot\cdot\cdot a_{m+n}^{k_{m+n}})=(-1)^{\hat{i}\hat{f_i}}~k_i~p_i^{\sum_{r=0}^{i-1}{z_rk_r}}~\prod_{r=1}^{i-1}{p_r^{-k_rz_i}}a_1^{k_1}a_2^{k_2}\cdot\cdot\cdot a_i^{k_i-1}\cdot\cdot\cdot a_{m+n}^{k_{m+n}} 
 \end{equation}
where $f_i=a_1^{k_1}\cdot a_2^{k_2}\cdot\cdot\cdot a_{i-1}^{k_{i-1}}$.
\begin{rem}
It is clearly seen from (\ref{18}) that the partial derivatives yield a representation of the superalgebra $\mathcal{A}$.
\end{rem}

Now we shall construct a quantum supergroup of the vector fields corresponding to Maurer-Cartan forms on ${\mathcal{A}}$. First, let us start by the right-invariant Maurer-Cartan formula for any $f\in {\mathcal{A}}$ \cite{SLW3}:
\begin{equation} \label{37}
w_f:=\mu(({\sf d}\otimes S_{\mathcal{A}}) \Delta_{\mathcal{A}}(f)),
\end{equation}
where $\mu$ stands for the multiplication. Thus we have 
\begin{equation} \label{38}
\begin{aligned}
 &\omega_{a_1} = {\sf d}a_1~a_1^{-1} \\
&\omega_{a_i} =  {\sf d}a_i~a_1^{-z_i}-z_i~{\sf d}a_1~a_1^{-1}a_i~a_1^{-z_i}, \quad i=2,3,...,m+n.
\end{aligned}
\end{equation}
We also need the commutation relations of the Maurer-Cartan forms $\omega_{a_i}$'s with the generators $x_i$'s: 
\begin{equation} \label{rel}
 a_i\omega_{a_j} = (-1)^{\hat{i}(\hat{j}+1)}p_j^{z_i}\omega_{a_j}a_i, \quad i,j=1,2,...,m+n. 
\end{equation}
Thus, taking into account the fact 
\begin{equation*}
{\sf d}:=\omega_{a_1}T_{a_1}+\omega_{a_2}T_{a_2}+...+\omega_{a_{m+n}}T_{a_{m+n}}\equiv {\sf d}a_1\partial_{a_1}+{\sf d}a_2\partial_{a_2}+...+{\sf d}a_{m+n}\partial_{a_{m+n}} 
\end{equation*}
where $T_{a_i}$'s are the vector fields corresponding to the Maurer-Cartan forms, we can write the vector fields in terms of the partial derivatives $\partial_{a_i}$'s:
 
\begin{equation*}
T_{a_1}= \sum_{k=1}^{m+n}{z_ia_i\partial_{a_i}}, \quad T_{a_i}=a_1^{z_i}\partial_{a_i}, \quad i=2,3,...,m+n.
\end{equation*}
To give the algebra of the vector fields, we compute the following super commutative algebra relations by using (\ref{10}), (\ref{17}) and (\ref{18}) as follows:
\begin{equation*}
T_{a_i}T_{a_j}=(-1)^{\hat{i}\hat{j}}T_{a_j}T_{a_i}, \quad i,j=1,2,...,m+n.
\end{equation*}
To construct Hopf superalgebra structure, it is sufficient to give the Leibniz rules related with the vector fields. For this, consider any monomial $f=a_1^{k_1}a_2^{k_2}\cdot\cdot\cdot a_{m+n}^{k_{m+n}}$ and any element $g$ of ${\mathcal{A}}$. In what follows we shall use the relation of $f$ with the Maurer-Cartan form $\omega_{a_i}$:
\begin{equation}\label{cart}
\begin{aligned}
&f\omega_{a_1}=(-1)^{\hat{f}}\omega_{a_1}f, \\
 &f\omega_{a_i}=(-1)^{\hat{f}\hat{\omega_{a_i}}}p_j^{\sum_{l=1}^{m+n}{z_lk_l}}\omega_{a_i}f, \quad i=2,3,...,m+n.
\end{aligned}
\end{equation}
Now we consider the action of ${\sf d}$ on $f\cdot g$:
\begin{equation*}
{\sf d}(f\cdot g)=(\omega_{a_1}T_{a_1}+\omega_{a_2}T_{a_2}+...+\omega_{a_{m+n}}T_{a_{m+n}})(f)\cdot g+
(-1)^{\hat{f}}f\cdot (\omega_{a_1}T_{a_1}+\omega_{a_2}T_{a_2}+...+\omega_{a_{m+n}}T_{a_{m+n}})(g).
\end{equation*}
To collect with respect to $\omega_{x_i}$'s, we use the relations given by (\ref{cart}):
\begin{equation*}
\begin{aligned}
(\omega_{a_1}T_{a_1}+\omega_{a_2}T_{a_2}+...+\omega_{a_{m+n}}T_{a_{m+n}})(f\cdot g)&= \omega_{a_1}(T_{a_1}(f)\cdot g+f\cdot T_{a_1}(g)) \\
&+\omega_{a_i}\left(T_{a_i}(f)\cdot g+(-1)^{\hat{f}\hat{a_i}}p_i^{{\sum_{l=1}^{m+n}{z_lk_l}}}f\cdot T_{a_i}(g)\right).
\end{aligned}
\end{equation*}
This last equation yields the following Leibniz rules:
\begin{equation*}
\begin{aligned}
&T_{a_1}(f\cdot g)=T_{a_1}(f)\cdot g+f\cdot T_{a_1}(g) \\
&T_{a_i}(f\cdot g)=T_{a_i}(f)\cdot g+(-1)^{\hat{f}\hat{a_i}}p_i^{{\sum_{l=1}^{m+n}{z_lk_l}}}f\cdot T_{a_i}(g), \quad i=2,3,...,m+n.
\end{aligned}
\end{equation*}
Thus using the above relation with the fact that the graded tensor product is $(X \otimes Y)(f \otimes g)=(-1)^{\hat{Y}\hat{f}}X(f)\otimes Y(g)$ and
$\mu \left(\Delta(X)(f \otimes g)\right):=X(f\cdot g)$ for any two elements $X$ and $Y$ with degrees $\hat{X}$ and $\hat{Y}$, respectively, we have the following deformed coproducts
for the quantum Lie superalgebra:
\begin{equation*}
\begin{aligned}
&\Delta(T_{a_1})=T_{a_1}\otimes 1 + 1 \otimes T_{a_1} \\
&\Delta(T_{a_i})=T_{a_i} \otimes 1+ p_i^{T_{a_1}} \otimes T_{a_i}, \quad i=2,3,...,m+n.
\end{aligned}
\end{equation*}
Note that we also use $T_{a_1}(f)=(\sum_{l=1}^{m+n}{z_lk_l})f$ obtained from (\ref{19}). From the Hopf algebra axioms, we also obtain the counit and the antipode as follows:
\begin{equation*}
\begin{aligned}
&\varepsilon(T_{a_i})=0, \quad i=1,2,...,m+n \\
&S(T_{a_1})=-T_{a_1}, \quad S(T_{a_i})=-p_i^{-T_{a_1}} T_{a_i}, \quad i=2,3,...,m+n.
\end{aligned}
\end{equation*}

\section{ Nonhomegeneous commutation relations derived from ${\mathcal{A}}$}

Let us generalize $\mathcal{A}$ to a new algebra obtained by considering formal series in the grouplike element $a_1$ such that
\begin{equation} \label{20}
\left(\sum c_k a_1^k \right) a_j= a_j \sum c_k (p_ja_1)^k, \quad c_k \in \mathbb{C},
\end{equation}
and the multiplication of two power series in $a_1$ is defined through the usual Cauchy product.  Now, set
\begin{equation} \label{21}
\begin{aligned}
x_1&:= ln(a_1), \quad e^{x_1}:=a_1  \\
x_i&:=a_1^{-1}a_i, \quad i=2,3,...,m+n \\
h_i&:=ln(p_i), \quad i=1,2,...,m+n
\end{aligned} 
\end{equation}
such that for $i,j=2,3,...,m+n$
\begin{equation} \label{22}
 [x_1,x_i]=h_i x_i,\quad x_i~x_j=(-1)^{\hat{i}\hat{j}}p_i^{1-z_j}p_j^{z_i-1}x_j~x_i, 
\end{equation}
where $ [u,v]= u v- v u$. Let $\mathcal{M}$ denote a new algebra
generated by $x_i$'s. From (\ref{22}), the following noncommutative relations are obtained :
\begin{equation} \label{23}
x_1^k ~x_j^l = x_j^l ~(x_1+ l h_j)^k, \quad
x_i^k~ x_j^l =\left((-1)^{\hat{i}\hat{j}}p_i^{1-z_j}p_j^{z_i-1}\right)^{kl}x_j^l~x_i^k
\end{equation}
for $ k,l \in {\mathbb{N}}$ and $i,j=2,3,...,m+n$.

\noindent Using Hopf superalgebra structure of the algebra $\mathcal{A}$, one can easily see that the coproduct for the algebra $\mathcal{M}$ appears as:
\begin{equation} \label{24}
\Delta_{\mathcal{M}}(x_i) =  e^{(z_i-1)x_1} \otimes x_i + x_i \otimes e^{(z_i-1)x_1}, \ \ i=1,2...,m+n
\end{equation}
the counit and the antipode are given as follows:
\begin{equation} \label{25}
\varepsilon_{\mathcal{M}}(x_i) = 0, \quad S_{\mathcal{M}}(x_i)= - e^{(1-z_i)x_1} ~x_i ~e^{(1-z_i)x_1} \ \ i=1,2,...,m+n.
\end{equation}
\begin{rem}
In fact, it is clear that the algebra $\mathcal{M}$ with commutation relations (\ref{22}) is a generalization of $\kappa-$ deformed superspace as the superspace extension of the $\kappa-$deformed Minkowski space \cite{PK,SZ}.
\end{rem}

Since $\mathcal{M}$ has a Hopf superalgebra structure, it is well known that there exists a bicovariant differential calculus over $\mathcal{M}$. In order to see this differential 
calculus explicitly, we want to see commutation parameters of the relevant relations in terms of the parameters $h_i$'s. To achieve this, one can use the following mappings having the properties of the right(left) bicovariant structure:
\begin{equation} \label{26}
\begin{aligned}
\Delta_R({\sf d}(f))&=(({\sf d} \otimes \mbox{id}) \circ \Delta_{\mathcal M})(f), \\
\Delta_L({\sf d}(f))&=((\mbox{id} \otimes {\sf d}) \circ \Delta_{\mathcal M})(f)
\end{aligned}
\end{equation}
for $f\in \mathcal{M}$. However, because of nonhomegeneous commutation relations (\ref{22}), we need to perform too much computation to realize the approach mentioned above. For this approach, the interested reader is referred to the bicovariant differential calculus on the $\kappa-$Minkowski space \cite{ASS}. Instead the approach used in \cite{ASS}, we use an approach that requires a simple computation, based on the following lemma \cite{MO2}:

\begin{lem} The following series is a representation of $ln (a_1)$
\begin{equation} \label{27}
\sum_{k=1}^{\infty}\dfrac{(-1)^{k+1}}{k} (a_1-1)^k,
\end{equation}
which means that it is consistent with relations (\ref{22}). Morever, based on the action of operator $\partial_{a_1}$ (see \ref{19}), this representation yields the fact that $\partial_{a_1}(lna_1)=a_1^{-1}.$
\end{lem}

\begin{thm}
We have the following relations of the generators $a,b$ and $\beta$ with their differentials, and these relations yield a bicovariant differential calculus over $\mathcal M$ 
\begin{equation} \label{28}
\begin{aligned}
&\left[x_1, {\sf d}x_1\right]=0, \\
&\left[x_1, {\sf d}x_i\right]=h_i {\sf d}x_i,~\left[x_i,{\sf d}x_j\right]_{\eta(i,j)}=0, \\
\end{aligned}
\end{equation}
where $~i=2,3,...,m+n,~j=1,2,...,m+n$, and  $[~,~]$ and $[~,~]_h$ are $\mathbb{Z}_2$-graded commutator and $q$-commutator defined by $[a,b]=ab-(-1)^{\hat{a}\hat{b}}ba$, $[a,b]_q=ab-(-1)^{\hat{a}\hat{b}}qba$, respectively, and $\eta(i,j)=p_i^{1-z_j}~p_j^{z_i-1}$. 
\end{thm}

\begin{proof}  
We first differentiate $ln(a_1)$, $a_1^{-1} a_i$, by using the relations (\ref{13}), which results in
\begin{equation} \label{29}
{\sf d}x_1={\sf d}a_1 a_1^{-1}, \quad {\sf d}x_i= h_i^{-1} {\sf d}a_i a_1^{-1} - h_i^{-1} {\sf d}a_1 a_1^{-1} a_i a_1^{-1},
\end{equation}
for $i=2,3,...,m+n$.
Using the series (\ref{27}), (\ref{29}) and (\ref{13}) in $ x_1{\sf d}x_i $ implies $x_1{\sf d}x_i= (-1)^{\hat{i}(\hat{j}+1)}{\sf d}x_ix_1+h_i{\sf d}x_i $. The other noncommutative relations are obtained in similar way. Now, we should prove that the relations (\ref{28}) yield a bicovariant differential calculus with (\ref{26}). Indeed, we first show that the mappings $\Delta_R$ and $\Delta_L$, acting on $\mathcal{M}$ as $\Delta_{\mathcal{M}}$ does, and the differentials as follows, preserve the relations (\ref{28}):

\begin{equation} \label{30}
\begin{aligned}
&\Delta_R({\sf d}x_i)  =  (z_i-1)e^{(z_i-1)x_1}{\sf d}x_1 \otimes x_i +{\sf d}x_i\otimes e^{(z_i-1)x_1}, \\   
&\Delta_L({\sf d}x_i)  =    e^{(z_i-1)x_1}\otimes{\sf d}x_i + (-1)^{\hat{i}}x_i\otimes (z_i-1)e^{(z_i-1)x_1}{\sf d}x_1.
\end{aligned}
\end{equation}

For this, for example, we show that
\begin{equation*}
\begin{aligned}
\Delta_L(x_1{\sf d}x_j)  =&(1 \otimes x_1 + x_1 \otimes 1)\left(e^{(z_j-1)x_1}\otimes{\sf d}x_j + (-1)^{\hat{j}}x_j\otimes (z_j-1)e^{(z_j-1)x_1}{\sf d}x_1\right) \\
   =& e^{(z_j-1)x_1} \otimes (h_j{\sf d}x_j+{\sf d}x_jx_1 ) + (-1)^{\hat{j}}x_j\otimes (z_j-1)x_1e^{(z_j-1)x_1}{\sf d}x_1+ x_1e^{(z_j-1)x_1}\otimes {\sf d}x_j \\
&+ (-1)^{\hat{j}} (h_j x_j+x_jx_1)\otimes(z_j-1)	e^{(z_j-1)x_1}{\sf d}x_1\\
	=& \left(e^{(z_j-1)x_1}\otimes{\sf d}x_j + (-1)^{\hat{j}}x_j\otimes (z_j-1)e^{(z_j-1)x_1}{\sf d}x_1\right) (1 \otimes x_1 + x_1 \otimes 1)\\
	&+ h_j \left(e^{(z_j-1)x_1}\otimes{\sf d}x_j + (-1)^{\hat{j}}x_j\otimes (z_j-1)e^{(z_j-1)x_1}{\sf d}x_1\right) \\
	=& \Delta_L({\sf d}x_j) \Delta_L(x_1)+ h_j\Delta_L({\sf d}x_j).
\end{aligned}
\end{equation*}
That is, $\Delta_L$ leaves invariant the commutation relation $\left[x_1, {\sf d}x_j\right]=h_j {\sf d}x_j,~j=2,3,...,m+n $.
It is also readily seen that the mappings given by (\ref{30}) hold the conditions of the bicovariant bimodule.
Morever, if we apply the exterior differential operator ${\sf d}$ to each relation in (\ref{28}), we have the following relations among the differentials 
\begin{equation} \label{32}
 {\sf d}x_i\wedge{\sf d}x_j=(-1)^{(\hat{i}+1)(\hat{j}+1)}\eta(i,j){\sf d}x_j\wedge{\sf d}x_i, \quad i,j= 1,2,...,m+n
\end{equation}
which is preserved under the mappings $\Delta_R,~\Delta_L$.
\end{proof}
\begin{cor}
The differential algebra with the relations (\ref{22}), (\ref{28}) and (\ref{32}) has a graded Hopf algebra structure induced by $\hat{\Delta}=\Delta_R+\Delta_L$.
\end{cor}

One can also easily obtain deformation relations between the 
operators and the generators of $ \mathcal{M} $ using the
Leibniz rule and the differential calculus of $ \mathcal{M} $ as follows
\begin{equation} \label{35}
\begin{aligned}
&\left[\partial_{x_1},x_1\right]=1, \quad \left[\partial_{x_i},x_1\right] = h_i\partial_{x_i},\quad i=2,3,...,m+n, \\
 &\left[\partial_{x_i},x_j\right]_{\eta(j,i)}=\delta_{ij},  \quad i=1,2,...,m+n, \quad j=2,3,...,m+n.
\end{aligned}
\end{equation}
Using the nilpotency rule ${\sf d}^2=0$, one gets commutation relations
\begin{equation} \label{36}
\left[\partial_{x_i},\partial_{x_j}\right]_{\eta(i,j)}=0, \quad i,j=1,2,...,m+n.
\end{equation}
Finally, the deformed Weyl superalgebra ${\mathbb{C}}\left\langle x_1,x_2,...,x_{m+n},\partial_{x_1},\partial_{x_2},...,\partial_{x_{m+n}} \right\rangle$ is given by the defining relations (\ref{22}), (\ref{35}) and (\ref{36}). This deformed Weyl superalgebra becomes the usual Weyl superalgebra when  the all parameters $p_1,p_2,...,p_{m+n} \rightarrow 0$.
                                                                                                                                 
\section{Maurer-Cartan 1-forms on $\mathcal{M}$}

\noindent The right-invariant Maurer-Cartan form corresponding to any $f\in \mathcal{M}$ can be given by the following formula \cite{SLW3}:
\begin{equation} \label{37}
w_f:=m(({\sf d}\otimes S_{\mathcal{M}}) \Delta_{\mathcal{M}}(f)),
\end{equation}
where $m$ stands for the multiplication. Thus we have 
\begin{equation} \label{38}
 \omega_{x_i} = [ {\sf d}x_i  + (1-z_i) {\sf d}x_1 x_i]\ e^{(1-z_i)x_1}, \quad i=1,2,...,m+n.
\end{equation}
Denote the algebra generated by 
$ \omega_{x_1}$, $\omega_{x_2}$,...,$\omega_{x_{m+n}}$ by $\Theta $. First determine all commutation relations about the Maurer-Cartan forms . 
The generators of  $\Theta $ and the generators of $\mathcal{M}$
satisfy the following rules
\begin{equation} \label{39}
\begin{aligned}
 &\left[x_1, \omega_{x_1}\right]=0, \quad \left[x_1, \omega_{x_i}\right]=h_i\omega_{x_i}, \\
&x_i\omega_{x_j}=(-1)^{\hat{i}(\hat{j}+1)}{\eta(i,j)}e^{h_i(z_j-1)}\omega_{x_j}x_i,
\end{aligned}
\end{equation}
where $ i=2,3,...,m+n, j=1,2,...,m+n$. 
The commutation rules of the generators of $\Theta$ are of the following form:
\begin{equation} \label{40}
\left[\omega_{x_i}, \omega_{x_j}\right]=0, \quad i,j=1,2,...,m+n.
\end{equation}
The algebra $\Theta$ is a graded Hopf algebra with the following
comappings: for $i=1,2,...,m+n$, the coproduct $\Delta_\Theta: \Theta \longrightarrow
\Theta \otimes \Theta$ is defined by
\begin{equation} \label{41}
\Delta_\Theta(\omega_{x_i}) = \omega_{x_i} \otimes 1 + 1 \otimes \omega_{x_i}, 
\end{equation}
The counit $\varepsilon_\Theta: \Theta \longrightarrow {\mathbb{C}}$ is
given by
\begin{equation} \label{42}
\varepsilon_\Theta(\omega_{x_i}) = 0, 
\end{equation}
and the antipode $S_\Theta: \Theta \longrightarrow \Theta$ is
defined by
\begin{equation} \label{43}
S_\Theta(\omega_{x_i}) = - \omega_{x_i}.
\end{equation}

\section{ Quantum Lie superalgebra of vector fields}

\noindent In this section we give vector fields corresponding to the Maurer-Cartan 1-forms on $\mathcal{M}$ and Lie superalgebra of the vector fields.
 First we rewrite the Maurer-Cartan forms as follows:
\begin{equation} \label{44}
{\sf d}x_i = \omega_{x_i}e^{(z_i-1)x_1}+(z_i-1){\sf d}x_1x_i, \quad i=1,2,...,m+n
\end{equation}
and consider the exterior differential ${\sf d}$ in the following form
\begin{equation} \label{45}
{\sf d} = \omega_{x_1} T_{x_1} + \omega_{x_2} T_{x_2} +...+ \omega_{x_{m+n}} T_{x_{m+n}},
\end{equation}
where $T_{x_i}$'s are vector fields corresponding to the Maurer-Cartan forms. We can determine the vector fields in terms of the partial derivative operators holding relations given in (\ref{35}) and (\ref{36}).
 By inserting (\ref{44}) to the expression
\begin{equation} \label{46}
{\sf d}={\sf d}x_1 \partial_{x_1} + {\sf d}x_2 \partial_{x_2} +...+ {\sf d}x_{m+n} \partial_{x_{m+n}}, 
\end{equation}
we can find the (quantum) Lie superalgebra generators expressed in terms
of the operators:
\begin{equation} \label{47}
\begin{aligned}
&T_{x_1} \equiv   \partial_{x_1}+ \sum_{i=2}^{m+n}{(z_i-1)x_i\partial_{x_i}}, \\ 
&T_{x_i} \equiv e^{(z_i-1)x_1}\partial_{x_i}, \quad i=2,3,...,m+n.
\end{aligned}
\end{equation}
 Now, we can obtain the commutation relations of these
generators, as follows, using (\ref{22}), (\ref{35}) and (\ref{36}):
\begin{equation} \label{48}
\left[T_{x_i}, T_{x_j}\right] =0, \quad i,j=1,2,...,m+n,
\end{equation}
where the parity of $T_{x_i}$ is the same with one of $x_i$.
The commutation relations in (\ref{48}) should be
compatible with monomials in $\mathcal{M}$. To realize this,
it is sufficient to get commutation relations between the generators of  Lie
superalgebra and the coordinates of $\mathcal{M}$. They are derived by (\ref{22}) and (\ref{35}) as
\begin{equation} \label{49}
\begin{aligned}
& \left[T_{x_1},x_1\right]=1,~\left[T_{x_i},x_1\right] =h_iT_{x_i}, ~i=2,3,...,m+n, \\
&\left[T_{x_i},x_j\right]_{\eta(j,i)e^{(z_j-1)h_j}} =e^{(z_i-1)x_1}\delta_{ij}, ~i=2,3,...,m+n,~j=2,3,...,m+n. \\
\end{aligned}
\end{equation}

\section { Conclusion }
We introduced a multiparametric quantum (m+n)-superspace which as algebra is noncommutative in the sense of Manin superplane, and as Hopf superalgebra is cocommutative. Then we see that a bicovariant differential calculus on this quantum superspace results in the partial derivatives which represent the superalgebra on this quantum superspace. Morever, we constructed a noncocommutative Hopf superalgebra(quantum supergroup) related with the vector fields corresponding to the Maurer-Cartan forms on this quantum (m+n)-superspace. We also define new elements with nonhomogeneous relations by the logarithm of the grouplike element in the quantum (m+n)-superspace with homogeneous relations. The most interesting part of this algebra induced from the quantum (m+n)-superspace is that it reduces to the $\kappa-$deformed Minkowski superspace by some convenient constrains on the deformation parameters $p_i$'s and arbitrary integers $z_i$'s.  Finally, a bicovariant differential calculus and the relevant results over the logarithmic extension of the quantum (m+n)-superspace are given. 
\section{Acknowledgement}
We are very grateful to the reviewers for their valuable and thoughtful comments.
%\bibliographystyle{amsplain}
%\bibliography{a-reference}

\end{document}